\documentclass[10pt,english,conference]{IEEEtran}
\usepackage[T1]{fontenc}
\usepackage[latin9]{inputenc}
\usepackage{float}
\usepackage{amsmath}
\usepackage{amsthm}
\usepackage{amssymb}
\usepackage{graphicx}

\makeatletter

\floatstyle{ruled}
\newfloat{algorithm}{tbp}{loa}
\providecommand{\algorithmname}{Algorithm}
\floatname{algorithm}{\protect\algorithmname}

\theoremstyle{plain}

\theoremstyle{plain}

\ifCLASSINFOpdf
\else
\fi
\usepackage{babel}

\makeatother

\usepackage{babel}
\providecommand{\propositionname}{Proposition}
\providecommand{\theoremname}{Theorem}

\begin{document}

\title{Fronthaul Quantization as Artificial Noise for Enhanced Secret Communication
in C-RAN}

\author{\IEEEauthorblockN{$^{1}$Seok-Hwan Park, \textit{Member}, \textit{IEEE}, $^{2}$Osvaldo Simeone, \textit{Fellow}, \textit{IEEE}, and $^{3}$Shlomo Shamai (Shitz), \textit{Fellow}, \textit{IEEE}} \IEEEauthorblockA{$^{1}$Division of Electronic Engineering, Chonbuk National University,
Jeonju-si, Jeollabuk-do, 54896 Korea\\
$^{2}$CWiP, New Jersey Institute of Technology, 07102 Newark, New
Jersey, USA\\
$^{3}$Department of Electrical Engineering, Technion, Haifa, 32000,
Israel\\
 Email: seokhwan@jbnu.ac.kr, osvaldo.simeone@njit.edu, sshlomo@ee.technion.ac.il}}
\maketitle
\begin{abstract}
This work considers the downlink of a cloud radio access network (C-RAN),
in which a control unit (CU) encodes confidential messages, each of
which is intended for a user equipment (UE) and is to be kept secret
from all the other UEs. As per the C-RAN architecture, the encoded
baseband signals are quantized and compressed prior to the transfer
to distributed radio units (RUs) that are connected to the CU via
finite-capacity fronthaul links. This work argues that the quantization
noise introduced by fronthaul quantization can be leveraged to act
as ``artificial'' noise in order to enhance the rates achievable
under secrecy constraints. To this end, it is proposed to control
the statistics of the quantization noise by applying multivariate,
or joint, fronthaul quantization/compression at the CU across all
outgoing fronthaul links. Assuming wiretap coding, the problem of
jointly optimizing the precoding and multivariate compression strategies,
along with the covariance matrices of artificial noise signals generated
by RUs, is formulated with the goal of maximizing the weighted sum
of achievable secrecy rates while satisfying per-RU fronthaul capacity
and power constraints. After showing that the artificial noise covariance
matrices can be set to zero without loss of optimaliy, an iterative
optimization algorithm is derived based on the concave convex procedure
(CCCP), and some numerical results are provided to highlight the advantages
of leveraging quantization noise as artificial noise.\end{abstract}


\theoremstyle{theorem}
\newtheorem{theorem}{Theorem}
\theoremstyle{proposition}
\newtheorem{proposition}{Proposition}
\theoremstyle{lemma}
\newtheorem{lemma}{Lemma}
\theoremstyle{corollary}
\newtheorem{corollary}{Corollary}
\theoremstyle{definition}
\newtheorem{definition}{Definition}
\theoremstyle{remark}
\newtheorem{remark}{Remark}

\section{Introduction\label{sec:Introduction}}

\let\thefootnote\relax\footnotetext{The work of S.-H. Park was supported by the NRF Korea funded by the Ministry of Science, ICT $\&$ Future Planning (MSIP) through grant 2015R1C1A1A01051825. The work of O. Simeone was partially supported by the U.S. NSF through grant 1525629. The work of S. Shamai has been supported by the European Union's Horizon 2020 Research And Innovation Programme, grant agreement no. 694630.}

Motivated by the original works on the wiretap channel \cite{Wyner}\cite{Csiszar-Korner},
\textit{\emph{physical-layer security techniques have been extensively
studied as effective means of protecting data secrecy for communications
over wireless channels in a variety of scenarios \cite{Liang-et-al:09}\cite{Bloch-Barros}.
One of the key techniques that have been devised for enhancing the
rates at which information can be transmitted securely is the addition
of artificial noise at the transmitter \cite{Goel}-\cite{Liu-et-al:TIT10}.
This strategy finds its theoretical justification in the prefix channel
approach that was shown in \cite{Csiszar-Korner} to achieve the secrecy
capacity of a general wiretap channel. }}

\begin{figure}
\centering\includegraphics[width=8cm,height=3.5cm]{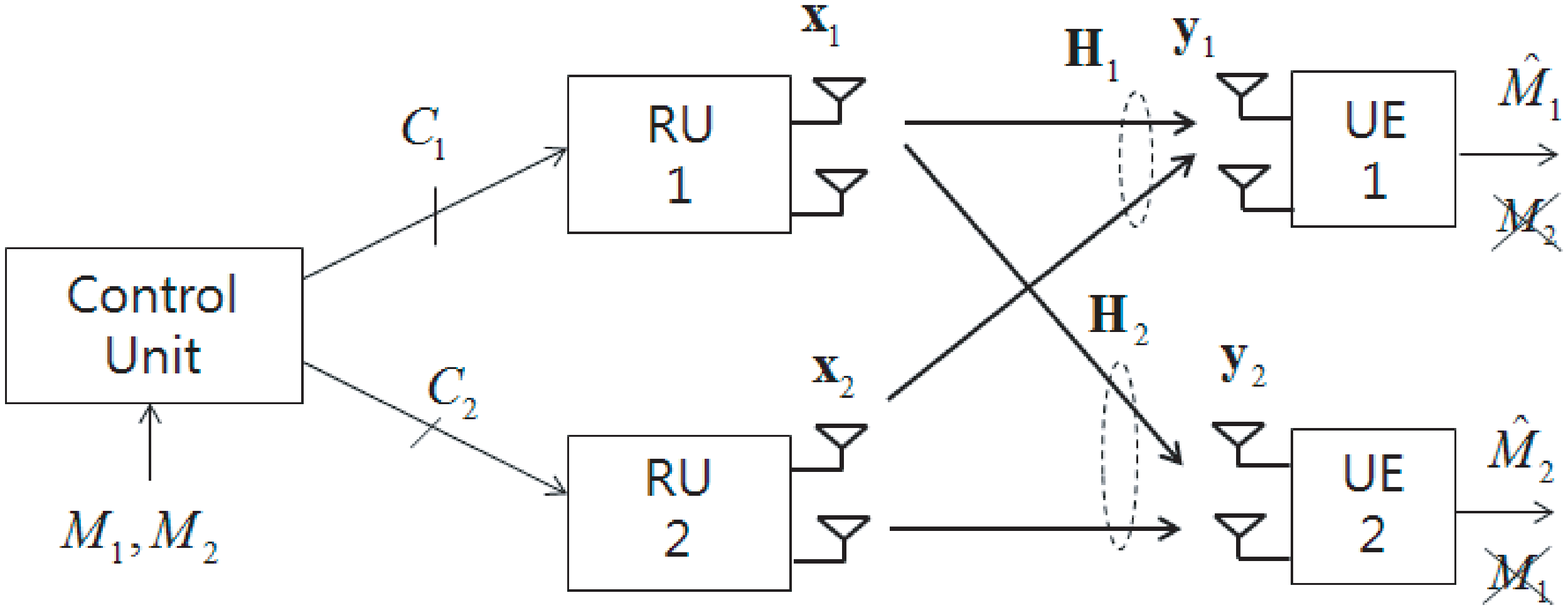}

\caption{\label{fig:System-Model}Illustration of the downlink of a C-RAN system
with confidential messages for $N_{R}=2$ RUs, $N_{U}=2$ UEs, $n_{R,i}=2$
RU antennas and $n_{U,k}=2$ UE antennas.}
\end{figure}

\textit{\emph{In this work, we study physical-layer secure communication
in the context of the downlink of cloud radio access networks}} (C-RANs).
In a C-RAN, a control unit (CU) performs joint encoding of the messages
intended for all the user equipments (UEs) located in the geographical
area covered by the radio units (RUs) connected to the CU. The encoded
baseband signals are then transferred to the RUs on the fronthaul
link in analog or digital format. For digital fronthauling, the CU
quantizes and compresses the encoded baseband signals prior to the
transfer to the RUs due to the limited bit rate of the fronthaul links
\cite{Park-et-al:TSP13}\cite{Simeone-et-al:JCN}.

In this work, we argue that the quantization noise introduced by fronthaul
quantization can be leveraged to act as artificial noise in order
to enhance the rates achievable under secrecy constraints. To this
end, as an extension of the work \cite{Park-et-al:TSP13}, we propose
to apply multivariate, or joint, fronthaul quantization/compression
\cite{ElGamal} at the CU for all outgoing fronthaul links in order
to control the statistics of the quantization noise. Multivariate
quantization/compression was recently shown in \cite{Park-et-al:TSP13}
to improve the performance of C-RANs without secrecy constraints with
respect to standard per-fronthaul link point-to-point quantization/compression.

In the rest of this paper, we first formulate the problem of jointly
optimizing precoding, multivariate compression and the covariance
matrices of artificial noise signals generated by RUs, with the goal
of maximizing the weighted sum of achievable secrecy rates of the
intended UEs subject to per-RU fronthaul capacity and power constraints
(Sec. \ref{sec:Problem-Definition-Linear}). We show that the artificial
noise covariance matrices can be set to zero with no loss of optimality
and hence focus on the joint optimization of precoding and multivariate
compression. An iterative optimization algorithm is then derived based
on the concave convex procedure (CCCP) (Sec. \ref{sec:Problem-Definition-Linear}),
and numerical evidence is provided to highlight the advantages of
the proposed schemes (Sec. \ref{sec:Numerical-Results}).

\section{System Model\label{sec:System-Model}}

We consider the downlink of a C-RAN in which a single CU controls
$N_{R}$ RUs. The CU communicates with the $i$th RU through a fronthaul
link of capacity $C_{i}$ bits/s/Hz where the normalization is with
respect to the downlink bandwidth. The $N_{R}$ RUs transmit signals
to $N_{U}$ UEs located in the union of the coverage areas of the
RUs. We define the sets $\mathcal{N}_{R}\triangleq\{1,\ldots,N_{R}\}$
and $\mathcal{N}_{U}\triangleq\{1,\ldots,N_{U}\}$ of the RUs and
the UEs, and denote the numbers of antennas of RU $i$ and UE $k$
as $n_{R,i}$ and $n_{U,k}$, respectively. An illustration is shown
in Fig. \ref{fig:System-Model}.

The signal $\mathbf{y}_{k}\in\mathbb{C}^{n_{U,k}\times1}$ received
by UE $k$ is given by
\begin{equation}
\mathbf{y}_{k}=\sum_{i\in\mathcal{N}_{R}}\mathbf{H}_{k,i}\mathbf{x}_{i}+\mathbf{z}_{k},\label{eq:received-signal-MS}
\end{equation}
where $\mathbf{H}_{k,i}\in\mathbb{C}^{n_{U,k}\times n_{R,i}}$ represents
the channel response matrix from RU $i$ to UE $k$; $\mathbf{x}_{i}\in\mathbb{C}^{n_{R,i}\times1}$
indicates the signal transmitted by RU $i$; and $\mathbf{z}_{k}\in\mathbb{C}^{n_{U,k}\times1}$
is the additive noise at UE $k$ distributed as $\mathbf{z}_{k}\sim\mathcal{CN}(\mathbf{0},\mathbf{\Sigma}_{\mathbf{z}_{k}})$.
Each RU $i$ is subject to the power constraint $\mathbb{E}\left\Vert \mathbf{x}_{i}\right\Vert ^{2}\leq P_{i}$.

The information messages $M_{k}\in\{1,\ldots,2^{nR_{k}}\}$, each
of rate $R_{k}$ bits/s/Hz, are encoded within a block of $n$ channel
uses, where $n$ is large enough to justify the use of information-theoretic
limits. The message $M_{k}$ is intended for UE $k\in\mathcal{N}_{U}$
and is required to be kept secret from the other UEs. Specifically,
we impose that the UEs $l\in\mathcal{N}_{U}\setminus\{k\}$ be unable
to decode the message $M_{k}$ even in the worst-case scenario where
they cooperate since the CU cannot control the activities of the UEs
(see \cite{Wyner}). The messages $\{M_{k}\}_{k\in\mathcal{N}_{U}}$
are processed by the CU in the two steps described in the following
subsections before being transferred to the RUs.

\subsection{Linear Precoding\label{sub:Precoding}}

The CU first encodes each message $M_{k}$ using a wiretap code \cite{Wyner}
and obtaining an encoded signal $\mathbf{s}_{k}\in\mathbb{C}^{d_{k}\times1}$,
which is distributed as $\mathbf{s}_{k}\sim\mathcal{CN}(\mathbf{0},\mathbf{I})$.
We assume that the number $d_{k}$ of data streams satisfies the condition
$d_{k}\leq\min\{n_{R},n_{U,k}\}$ with the notation $n_{R}\triangleq\sum_{i\in\mathcal{N}_{R}}n_{R,i}$.
In order to enable the management of the inter-UE interference and
to enhance secrecy, the CU performs linear precoding, or beamforming,
with a precoding matrix $\mathbf{A}\in\mathbb{C}^{n_{R}\times d}$,
yielding the precoded signal $\tilde{\mathbf{x}}=[\tilde{\mathbf{x}}_{1};\ldots;\tilde{\mathbf{x}}_{N_{R}}]\in\mathbb{C}^{n_{R}\times1}$
with
\begin{equation}
\tilde{\mathbf{x}}=\mathbf{A}\mathbf{s},\label{eq:precoding}
\end{equation}
where $\tilde{\mathbf{x}}_{i}\in\mathbb{C}^{n_{R,i}\times1}$ is the
signal to be communicated to RU $i$; $\mathbf{s}\triangleq[\mathbf{s}_{1};\ldots;\mathbf{s}_{N_{U}}]$
is the vector of the signals encoded for the UEs; and we have defined
the notation $d\triangleq\sum_{k\in\mathcal{N}_{U}}d_{k}$. We note
that the discussion can be easily extended to systems, in which the
CU performs non-linear secrecy dirty-paper coding (S-DPC) precoding
proposed in \cite{Liu-et-al:TIT10}.

\subsection{Fronthaul Compression\label{sub:Fronthaul-Compression}}

The precoded baseband signal $\tilde{\mathbf{x}}_{i}$ needs to be
compressed prior to transmission to the RU, since the CU communicates
to RU $i$ through a fronthaul link of capacity $C_{i}$ bits/s/Hz.
Using standard rate-distortion considerations, we model the impact
of compression by adding a quantization noise $\mathbf{q}_{i}$ to
the compression input signal $\tilde{\mathbf{x}}_{i}$ so that the
compression output signal $\hat{\mathbf{x}}_{i}$ is given as
\begin{equation}
\hat{\mathbf{x}}_{i}=\tilde{\mathbf{x}}_{i}+\mathbf{q}_{i},\label{eq:compression-model}
\end{equation}
where the quantization noise $\mathbf{q}_{i}$ is independent of the
signal $\tilde{\mathbf{x}}_{i}$ and is distributed as $\mathbf{q}_{i}\sim\mathcal{CN}(\mathbf{0},\mathbf{\Omega}_{i,i})$.
Each RU $i$ decompresses the baseband signal $\hat{\mathbf{x}}_{i}$
based on the bit stream received on the fronthaul link. We emphasize
that, as done in, e.g., \cite{Park-et-al:TSP13}, quantization is
not designed so as to minimize the (e.g., quadratic) distortion between
the precoded signals $\tilde{\mathbf{x}}_{i}$ and the compressed
signals $\hat{\mathbf{x}}_{i}$, but rather with the aim of maximizing
the weighted sum of achievable secrecy rates, which will be defined
in Sec. \ref{sub:Achievable-Rates}.

In the standard point-to-point compression approach \cite{Simeone-et-al:ETT09},
in which the precoded signals $\tilde{\mathbf{x}}_{i}$ and $\tilde{\mathbf{x}}_{j}$
for different RUs $i\neq j$ are separately compressed, the quantization
noises $\mathbf{q}_{i}$ and $\mathbf{q}_{j}$ are independent, i.e.,
$\mathbf{\Omega}_{i,j}=\mathbf{0}$ for $i\neq j$. Instead, multivariate,
or joint, compression \cite[Ch. 9]{ElGamal} allows the CU to correlate
the quantization noises $\mathbf{q}_{1},\ldots,\mathbf{q}_{N_{R}}$
by jointly compressing the signals $\tilde{\mathbf{x}}_{1},\ldots,\tilde{\mathbf{x}}_{N_{R}}$.
This adds a further degree of freedom to the system design, which
will be leveraged here to enhance physical-layer security. It was
shown in \cite[Sec. IV-D]{Park-et-al:TSP13} that multivariate compression
can be implemented with no loss of optimality using a low-complexity
sequential processing architecture.

Specifically, in this work, we propose to shape the quantization noise
signals in order to enhance the secrecy performance by controlling
the correlation matrix $\mathbf{\Omega}$ of the quantization noise
vector $\mathbf{q}\triangleq[\mathbf{q}_{1};\ldots;\mathbf{q}_{N_{R}}]$,
where the covariance matrix $\mathbf{\Omega}\triangleq\mathbb{E}[\mathbf{q}\mathbf{q}^{\dagger}]$
is given as a block matrix whose $(i,j)$th block is $\mathbf{\Omega}_{i,j}\triangleq\mathbb{E}[\mathbf{q}_{i}\mathbf{q}_{j}^{\dagger}]$.
As mentioned, this control can be realized by means of multivariate
compression, which was recently demonstrated in \cite{Park-et-al:TSP13}
to achieve performance gains in terms of non-secrecy information rates.

It is a classic result in network information theory that the quantized
signals (\ref{eq:compression-model}) with the given quantization
noise covariance $\mathbf{\Omega}$ can be recovered by the RUs if
the conditions
\begin{align}
g_{\mathcal{S}}\left(\mathbf{A},\mathbf{\Omega}\right)\triangleq & \sum_{i\in\mathcal{S}}h\left(\mathbf{x}_{i}\right)-h\left(\hat{\mathbf{x}}_{\mathcal{S}}|\tilde{\mathbf{x}}\right)\nonumber \\
= & \sum_{i\in\mathcal{S}}\log_{2}\det\left(\mathbf{E}_{i}^{\dagger}(\mathbf{A}\mathbf{A}^{\dagger}+\mathbf{\Omega})\mathbf{E}_{i}\right)\nonumber \\
- & \log_{2}\det\left(\mathbf{E}_{\mathcal{S}}^{\dagger}\mathbf{\Omega}\mathbf{E}_{\mathcal{S}}\right)\leq\sum_{i\in\mathcal{S}}C_{i}\label{eq:fronthaul-constraint-multivariate}
\end{align}
are satisfied for all subsets $\mathcal{S}\subseteq\mathcal{N}_{R}$,
where we have defined the set $\hat{\mathbf{x}}_{\mathcal{S}}\triangleq\{\hat{\mathbf{x}}_{i}\}_{i\in\mathcal{S}}$
and the matrix $\mathbf{E}_{\mathcal{S}}$ obtained by stacking the
matrices $\mathbf{E}_{i}$ for $i\in\mathcal{S}$ horizontally with
the matrices $\mathbf{E}_{i}\in\mathbb{C}^{n_{R}\times n_{R,i}}$
having all-zero elements except for the rows from $(\sum_{j=1}^{i-1}n_{R,j}+1)$
to $(\sum_{j=1}^{i}n_{R,j})$ being the identity matrix of size $n_{R,i}$
\cite[Ch. 9]{ElGamal}.

\subsection{Artificial Noise\label{sub:Artificial-Noise}}

Based on the decompressed baseband signal $\hat{\mathbf{x}}_{i}$,
each RU $i$ creates the signal $\mathbf{x}_{i}$ to be transmitted
in the downlink as
\begin{equation}
\mathbf{x}_{i}=\hat{\mathbf{x}}_{i}+\mathbf{n}_{i},\label{eq:adding-artificial-noise}
\end{equation}
where $\mathbf{n}_{i}$ represents the artificial noise signal generated
by RU $i$ and is distributed as $\mathbf{n}_{i}\sim\mathcal{CN}(\mathbf{0},\mathbf{\Phi}_{i})$.
The artificial noise signals $\mathbf{n}_{i}$ are independent across
the index $i$ since each signal $\mathbf{n}_{i}$ is locally produced
by the corresponding RU. As for the quantization noise signal $\mathbf{q}_{i}$,
we need to carefully design the covariance matrix $\mathbf{\Phi}_{i}$
based on the channel matrices in order to enhance secrecy.

\subsection{Achievable Secrecy Rates\label{sub:Achievable-Rates}}

The signal $\mathbf{y}_{k}$ in (\ref{eq:received-signal-MS}) received
by UE $k$ can be written as
\begin{align}
\mathbf{y}_{k}=\!\mathbf{H}_{k}\mathbf{A}_{k}\mathbf{s}_{k}+\!\!\sum_{l\in\mathcal{N}_{U}\setminus\{k\}}\!\!\mathbf{H}_{k}\mathbf{A}_{l}\mathbf{s}_{l}\!+\!\mathbf{H}_{k}\!\left(\mathbf{q}+\mathbf{n}\!\right)\!+\mathbf{z}_{k},\label{eq:received-signal-rewritten}
\end{align}
where we defined the channel matrix $\mathbf{H}_{k}\triangleq[\mathbf{H}_{k,1}\,\ldots\,\mathbf{H}_{k,N_{R}}]$
from all the RUs to UE $k$, the aggregate vector $\mathbf{n}\triangleq[\mathbf{n}_{1};\ldots;\mathbf{n}_{N_{R}}]$
of the artificial noise signals and the submatrix $\mathbf{A}_{k}\in\mathbb{C}^{n_{R}\times d_{k}}$
of $\mathbf{A}$ multiplied to the signal $\mathbf{s}_{k}$ encoded
for UE $k$. The first term in (\ref{eq:received-signal-rewritten})
indicates the desired signal to be decoded by the receiving UE $k$,
the second term represents the inter-UE interference signals, which
encode the unintended messages, and the third and last terms are channelized
quantization noise and antenna additive noise signals, respectively.
Eq. (\ref{eq:received-signal-rewritten}) suggests that a joint design
of $\mathbf{A}$ and $\mathbf{\Omega}$ has the potential to jointly
``shape'' the useful signals and the quantization noise signals
to enhance the secrecy rate.

Assuming that each UE $k$ decodes the message $M_{k}$ based on the
signal $\mathbf{y}_{k}$ in (\ref{eq:received-signal-rewritten})
while treating the interference signals as noise, it was shown in
\cite{Wyner} that the rate
\begin{align}
 & R_{k}=\left[f_{k}\left(\mathbf{A},\mathbf{\Omega},\mathbf{\Phi}\right)\right]^{+}\label{eq:achievable-rate-1}
\end{align}
is achievable for UE $k$ ensuring that the other UEs cannot decode
the message $M_{k}$, where we defined the function
\begin{align}
 & f_{k}\left(\mathbf{A},\mathbf{\Omega},\mathbf{\Phi}\right)\triangleq I\left(\mathbf{s}_{k};\mathbf{y}_{k}\right)-I\left(\mathbf{s}_{k};\mathbf{y}_{\bar{k}}\right)\label{eq:achievable-rate-mutual-information}\\
 & =\phi\left(\!\mathbf{H}_{k}\mathbf{R}_{l}\mathbf{H}_{k}^{\dagger},\sum_{l\in\mathcal{N}_{U}\setminus\{k\}}\!\!\mathbf{H}_{k}\mathbf{R}_{l}\mathbf{H}_{k}^{\dagger}+\!\mathbf{H}_{k}(\mathbf{\Omega}+\mathbf{\Phi})\mathbf{H}_{k}^{\dagger}+\!\mathbf{\Sigma}_{\mathbf{z}_{k}}\!\right)\nonumber \\
 & -\phi\left(\!\mathbf{H}_{\bar{k}}\mathbf{R}_{k}\mathbf{H}_{\bar{k}}^{\dagger},\sum_{l\in\mathcal{N}_{U}\setminus\{k\}}\!\!\mathbf{H}_{\bar{k}}\mathbf{R}_{l}\mathbf{H}_{\bar{k}}^{\dagger}+\!\mathbf{H}_{\bar{k}}(\mathbf{\Omega}+\mathbf{\Phi})\mathbf{H}_{\bar{k}}^{\dagger}+\!\mathbf{\Sigma}_{\mathbf{z}_{\bar{k}}}\!\right),\nonumber
\end{align}
with the functions $\phi(\mathbf{A},\mathbf{B})\triangleq\log_{2}\det(\mathbf{A}+\mathbf{B})-\log_{2}\det(\mathbf{B})$
and $[x]^{+}=\max(0,x)$ and the notations $\mathbf{\Phi}\triangleq\mathrm{diag}(\{\mathbf{\Phi}_{i}\}_{i\in\mathcal{N}_{R}})$
and $\mathbf{R}_{k}\triangleq\mathbf{A}_{k}\mathbf{A}_{k}^{\dagger}$.
The vector $\mathbf{y}_{\bar{k}}$, defined as
\begin{equation}
\mathbf{y}_{\bar{k}}\triangleq[\mathbf{y}_{1};\ldots;\mathbf{y}_{k-1};\mathbf{y}_{k+1};\ldots;\mathbf{y}_{N_{U}}]=\mathbf{H}_{\bar{k}}\mathbf{x}+\mathbf{z}_{\bar{k}},
\end{equation}
represents the vector obtained by stacking the signals $\mathbf{y}_{l}$
received by the malicious UEs $l\in\mathcal{N}_{U}\setminus\{k\}$,
where we have defined the notations $\mathbf{H}_{\bar{k}}\triangleq[\mathbf{H}_{1}^{\dagger}\ldots\mathbf{H}_{k-1}^{\dagger}\mathbf{H}_{k+1}^{\dagger}\ldots\mathbf{H}_{N_{U}}^{\dagger}]^{\dagger}$
and $\mathbf{z}_{\bar{k}}\triangleq[\mathbf{z}_{1}^{\dagger}\ldots\mathbf{z}_{k-1}^{\dagger}\mathbf{z}_{k+1}^{\dagger}\ldots\mathbf{z}_{N_{U}}^{\dagger}]^{\dagger}$.
The maximization of the secrecy sum-rate over the quantization noise
covariance matrix $\mathbf{\Omega}$ entails that the rate loss induced
by the quantization noise is minimized at the intended UE while it
is maximized at the unintended UEs.

\section{Problem Definition and Optimization\label{sec:Problem-Definition-Linear}}

We aim at optimizing the precoding matrix $\mathbf{A}$, the quantization
noise covariance matrix $\mathbf{\Omega}$ and the artificial noise
covariance matrix $\mathbf{\Phi}$ with the goal of maximizing the
weighted sum of secrecy rates subject to the per-RU power and the
fronthaul capacity constraints. The problem is stated as\begin{subequations}\label{eq:problem-linear}
\begin{align}
\underset{\mathbf{A},\{\mathbf{\Omega},\mathbf{\Phi}\succeq\mathbf{0}\}}{\mathrm{maximize}} & \sum_{k\in\mathcal{N}_{U}}w_{k}\left[f_{k}\left(\mathbf{A},\mathbf{\Omega},\mathbf{\Phi}\right)\right]^{+}\label{eq:objective-linear}\\
\mathrm{s.t.}\,\, & g_{\mathcal{S}}\left(\mathbf{A},\mathbf{\Omega}\right)\leq\sum_{i\in\mathcal{S}}C_{i},\,\,\mathrm{for\,\,all}\,\,\mathcal{S}\subseteq\mathcal{N}_{R},\label{eq:constraint-fronthaul-linear}\\
 & \mathrm{tr}\left(\mathbf{E}_{i}^{\dagger}\mathbf{A}\mathbf{A}^{\dagger}\mathbf{E}_{i}+\mathbf{\Omega}_{i,i}+\mathbf{\Phi}_{i}\right)\nonumber \\
 & \,\,\,\,\,\,\,\,\,\leq P_{i},\,\,\mathrm{for\,\,all}\,\,i\in\mathcal{N}_{R}.\label{eq:constraint-power-linear}
\end{align}
\end{subequations}The following lemma shows that we can reduce the
optimization domain without loss of optimality.

\begin{lemma}\label{lem:no-artificial-noise}Setting $\mathbf{\Phi}=\mathbf{0}$
in the problem (\ref{eq:problem-linear}), which corresponds to adding
no artificial noise at the RUs, does not cause any loss of optimality.

\end{lemma}

\begin{proof}Suppose that an optimal solution $(\mathbf{A}^{*},\mathbf{\Omega}^{*},\mathbf{\Phi}^{*})$
exists with $\mathbf{\Omega}^{*}\neq\mathbf{0}$. We can then define
another solution given by $(\mathbf{A}^{*},\mathbf{\Omega}^{*}+\mathbf{\Phi}^{*},\mathbf{0})$
that achieves exactly the same objective (\ref{eq:objective-linear})
without violating any of the constraints. This is because the left-hand
side of the fronthaul capacity constraint (\ref{eq:fronthaul-constraint-multivariate})
is non-increasing with respect to the covariance matrix $\mathbf{\Omega}$,
that is, adding quantization noise can only alleviate the fronthaul
overhead and hence any artificial noise added by the RUs can be directly
added to the quantization noise without loss of optimality.\end{proof}

From Lem. \ref{lem:no-artificial-noise}, we set $\mathbf{\Phi}=\mathbf{0}$
without loss of optimality. However, it is still not easy to solve
the problem (\ref{eq:problem-linear}) due to the non-smoothness (and
non-convexity) of the objective function. In order to make the problem
more tractable, we propose to solve an alternative problem obtained
by replacing the objective function with a smooth function
\begin{equation}
\sum_{k\in\mathcal{N}_{U}}w_{k}f_{k}\left(\mathbf{A},\mathbf{\Omega}\right),\label{eq:objective-LB}
\end{equation}
where we remove the dependence on the covariance $\mathbf{\Phi}=\mathbf{0}$.
Then, solving the obtained problem with respect to the variables $\mathbf{R}\triangleq\{\mathbf{R}_{k}\}_{k\in\mathcal{N}_{U}}$
and $\mathbf{\Omega}$ is a difference-of-convex (DC) program, and
we can adopt an iterative algorithm based on the CCCP as in \cite{Park-et-al:TSP13}.
The detailed algorithm is described in Algorithm I, where we defined
the functions $\tilde{f}_{k}(\{\mathbf{R}^{(t+1)},\mathbf{\Omega}^{(t+1)}\},\{\mathbf{R}^{(t)},\mathbf{\Omega}^{(t)}\})$
and $\tilde{g}_{\mathcal{S}}(\{\mathbf{R}^{(t+1)},\mathbf{\Omega}^{(t+1)}\},\{\mathbf{R}^{(t)},\mathbf{\Omega}^{(t)}\})$
as
\begin{align}
 & \tilde{f}_{k}\left(\{\mathbf{R}^{(t+1)},\mathbf{\Omega}^{(t+1)}\},\{\mathbf{R}^{(t)},\mathbf{\Omega}^{(t)}\}\right)\label{eq:linearized-objective}\\
\triangleq & \log_{2}\det\left(\sum_{l\in\mathcal{N}_{U}}\mathbf{H}_{k}\mathbf{R}_{l}^{(t+1)}\mathbf{H}_{k}^{\dagger}+\mathbf{H}_{k}\mathbf{\Omega}^{(t+1)}\mathbf{H}_{k}^{\dagger}+\mathbf{\Sigma}_{\mathbf{z}_{k}}\right)\nonumber \\
- & \log_{2}\det\left(\sum_{l\in\mathcal{N}_{U}\setminus\{k\}}\mathbf{H}_{k}\mathbf{R}_{l}^{(t)}\mathbf{H}_{k}^{\dagger}+\mathbf{H}_{k}\mathbf{\Omega}^{(t)}\mathbf{H}_{k}^{\dagger}+\mathbf{\Sigma}_{\mathbf{z}_{k}}\right)\nonumber \\
- & \log_{2}\det\left(\sum_{l\in\mathcal{N}_{U}}\mathbf{H}_{\bar{k}}\mathbf{R}_{l}^{(t)}\mathbf{H}_{\bar{k}}^{\dagger}+\mathbf{H}_{\bar{k}}\mathbf{\Omega}^{(t)}\mathbf{H}_{\bar{k}}^{\dagger}+\mathbf{\Sigma}_{\mathbf{z}_{\bar{k}}}\right)\nonumber \\
+ & \log_{2}\det\left(\sum_{l\in\mathcal{N}_{U}\setminus\{k\}}\mathbf{H}_{\bar{k}}\mathbf{R}_{l}^{(t+1)}\mathbf{H}_{\bar{k}}^{\dagger}+\mathbf{H}_{\bar{k}}\mathbf{\Omega}^{(t+1)}\mathbf{H}_{\bar{k}}^{\dagger}+\mathbf{\Sigma}_{\mathbf{z}_{\bar{k}}}\right)\nonumber \\
- & \varphi\left(\begin{array}{c}
\sum_{l\in\mathcal{N}_{U}\setminus\{k\}}\mathbf{H}_{k}\mathbf{R}_{l}^{(t+1)}\mathbf{H}_{k}^{\dagger}+\mathbf{H}_{k}\mathbf{\Omega}^{(t+1)}\mathbf{H}_{k}^{\dagger}+\mathbf{\Sigma}_{\mathbf{z}_{k}},\\
\sum_{l\in\mathcal{N}_{U}\setminus\{k\}}\mathbf{H}_{k}\mathbf{R}_{l}^{(t)}\mathbf{H}_{k}^{\dagger}+\mathbf{H}_{k}\mathbf{\Omega}^{(t)}\mathbf{H}_{k}^{\dagger}+\mathbf{\Sigma}_{\mathbf{z}_{k}}
\end{array}\right)\nonumber \\
- & \varphi\left(\begin{array}{c}
\sum_{l\in\mathcal{N}_{U}}\mathbf{H}_{\bar{k}}\mathbf{R}_{l}^{(t+1)}\mathbf{H}_{\bar{k}}^{\dagger}+\mathbf{H}_{\bar{k}}\mathbf{\Omega}^{(t+1)}\mathbf{H}_{\bar{k}}^{\dagger}+\mathbf{\Sigma}_{\mathbf{z}_{\bar{k}}},\\
\sum_{l\in\mathcal{N}_{U}}\mathbf{H}_{\bar{k}}\mathbf{R}_{l}^{(t)}\mathbf{H}_{\bar{k}}^{\dagger}+\mathbf{H}_{\bar{k}}\mathbf{\Omega}^{(t)}\mathbf{H}_{\bar{k}}^{\dagger}+\mathbf{\Sigma}_{\mathbf{z}_{\bar{k}}}
\end{array}\right),\nonumber
\end{align}
and
\begin{align}
 & \tilde{g}_{\mathcal{S}}\left(\{\mathbf{R}^{(t+1)},\mathbf{\Omega}^{(t+1)}\},\{\mathbf{R}^{(t)},\mathbf{\Omega}^{(t)}\}\right)\label{eq:linearized-fronthaul-constraint}\\
\triangleq & \sum_{i\in\mathcal{S}}\log_{2}\det\left(\mathbf{E}_{i}^{\dagger}(\sum_{k\in\mathcal{N}_{U}}\mathbf{R}_{k}^{(t)}+\mathbf{\Omega}^{(t)})\mathbf{E}_{i}\right)\nonumber \\
 & -\log_{2}\det\left(\mathbf{E}_{\mathcal{S}}^{\dagger}\mathbf{\Omega}^{(t+1)}\mathbf{E}_{\mathcal{S}}\right)\nonumber \\
 & +\sum_{i\in\mathcal{S}}\varphi\left(\begin{array}{c}
\sum_{k\in\mathcal{N}_{U}}\mathbf{E}_{i}^{\dagger}\mathbf{R}_{k}^{(t+1)}\mathbf{E}_{i}+\mathbf{E}_{i}^{\dagger}\mathbf{\Omega}^{(t+1)}\mathbf{E}_{i},\\
\sum_{k\in\mathcal{N}_{U}}\mathbf{E}_{i}^{\dagger}\mathbf{R}_{k}^{(t)}\mathbf{E}_{i}+\mathbf{E}_{i}^{\dagger}\mathbf{\Omega}^{(t)}\mathbf{E}_{i}
\end{array}\right)\nonumber
\end{align}
with the definition $\varphi(\mathbf{X},\mathbf{Y})\triangleq\log\det\left(\mathbf{Y}\right)+\mathrm{tr}(\mathbf{Y}^{-1}(\mathbf{X}-\mathbf{Y}))/\ln2$.

After convergence of Algorithm 1, the actual precoding matrix $\mathbf{A}_{k}$
for UE $k$ is obtained via rank reduction as $\mathbf{A}_{k}\leftarrow\mathbf{V}_{k}\mathbf{D}_{k}^{1/2}$,
where $\mathbf{D}_{k}$ is a diagonal matrix whose diagonal elements
are the $d_{k}$ leading eigenvalues of $\mathbf{R}_{k}^{(t+1)}$
and the columns of $\mathbf{V}_{k}$ are the corresponding eigenvectors.
This transformation from $\mathbf{R}_{k}^{(t+1)}$ to $\mathbf{A}_{k}$
may cause suboptimality (in terms of local optima) when the rank of
the matrix $\mathbf{R}_{k}^{(t+1)}$ is larger than $d_{k}$. However,
note that the obtained precoding matrices $\{\mathbf{A}_{k}\}_{k\in\mathcal{N}_{U}}$
together with $\mathbf{\Omega}^{(t+1)}$ are feasible in that they
satisfy the conditions (\ref{eq:constraint-fronthaul-linear}) and
(\ref{eq:constraint-power-linear}), since the matrices $\{\mathbf{R}_{k}^{(t+1)}\}_{k\in\mathcal{N}_{U}}$
satisfy (\ref{eq:constraint-fronthaul-linear}) and (\ref{eq:constraint-power-linear})
at each iteration. We also mention that, when the matrices $\mathbf{A}$
and $\mathbf{\Omega}$ are optimized upon as per problem (\ref{eq:problem-linear}),
we necessarily have that $f_{k}\left(\mathbf{A},\mathbf{\Omega}\right)\geq0$.
The reason is that it is always possible to obtain $f_{k}\left(\mathbf{A},\mathbf{\Omega}\right)=0$
by setting $\mathbf{A}_{k}=\mathbf{0}$ in order to satisfy the constraints
(\ref{eq:constraint-fronthaul-linear}) and (\ref{eq:constraint-power-linear}).

\begin{algorithm}
\caption{DC programming algorithm for problem (\ref{eq:problem-linear})}

1. Initialize the matrices $\mathbf{R}^{(1)}$ and $\mathbf{\Omega}^{(1)}$
to arbitrary feasible positive semidefinite matrices for problem (\ref{eq:problem-linear})
and set $t=1$.

2. Update the matrices $\mathbf{R}^{(t+1)}$ and $\mathbf{\Omega}^{(t+1)}$
as a solution of the convex problem
\begin{align*}
\underset{\mathbf{R}^{(t+1)},\mathbf{\Omega}^{(t+1)}\succeq\mathbf{0}}{\mathrm{maximize}} & \sum_{k\in\mathcal{N}_{U}}w_{k}\tilde{f}_{k}\left(\{\mathbf{R}^{(t+1)},\mathbf{\Omega}^{(t+1)}\},\{\mathbf{R}^{(t)},\mathbf{\Omega}^{(t)}\}\right)\\
\mathrm{s.t.}\,\, & \tilde{g}_{\mathcal{S}}\left(\{\mathbf{R}^{(t+1)},\mathbf{\Omega}^{(t+1)}\},\{\mathbf{R}^{(t)},\mathbf{\Omega}^{(t)}\}\right)\\
 & \,\,\,\,\leq\sum_{i\in\mathcal{S}}C_{i},\,\,\mathrm{for\,\,all}\,\,\mathcal{S}\subseteq\mathcal{N}_{R},\\
 & \sum_{k\in\mathcal{N}_{U}}\mathrm{tr}\left(\mathbf{E}_{i}^{\dagger}\mathbf{R}_{k}^{(t+1)}\mathbf{E}_{i}\right)+\mathrm{tr}(\mathbf{\Omega}_{i,i}^{(t+1)})\\
 & \,\,\,\,\leq P_{i},\,\,\mathrm{for\,\,all}\,\,i\in\mathcal{N}_{R}.
\end{align*}

3. Stop if a convergence criterion is satisfied. Otherwise, set $t\leftarrow t+1$
and go back to Step 2.
\end{algorithm}

\section{Numerical Results\label{sec:Numerical-Results}}

In this section, we present numerical results to validate the effectiveness
of the proposed secure design based on multivariate compression. We
compare four different strategies, i.e., non-secure and secure design
based on point-to-point and multivariate fronthaul compression strategies.
For the non-secure design, the problem (\ref{eq:problem-linear})
is tackled without taking into account the second term in (\ref{eq:achievable-rate-mutual-information})
that represents the penalty for guaranteeing the security. The so-obtained
precoding and quantization noise covariance matrices are then used
in (\ref{eq:achievable-rate-1}) to evaluate the secrecy sum-rate.
Unless stated otherwise, we focus on evaluating the average secrecy
sum-rate performance given in (\ref{eq:objective-linear}). We assume
that the locations of RUs and UEs are sampled from a uniform distribution
within a square area of side length $500$m, and the channel matrices
$\mathbf{H}_{k,i}$ are modeled as $\mathbf{H}_{k,i}=\sqrt{\gamma_{k,i}}\mathbf{H}_{k,i}^{w}$,
where the path-loss $\gamma_{k,i}$ is obtained as $\gamma_{k,i}=1/(1+(d_{k,i}/d_{0})^{\alpha})$
with $\alpha$, $d_{k,i}$ and $d_{0}$ being the path-loss exponent,
the distance between RU $i$ and UE $k$ and the reference distance,
respectively, and the elements of the channel matrices $\mathbf{H}_{k,i}^{w}$
are independent and identically distributed (i.i.d.) as $\mathcal{CN}(0,1)$.
We also assume that $P_{i}=P$, $C_{i}=C$ and $\mathbf{\Sigma}_{\mathbf{z}_{k}}=\mathbf{I}$
for all $i\in\mathcal{N}_{R}$ and $k\in\mathcal{N}_{U}$, and $\alpha=3$
and $d_{0}=50$m.

\begin{figure}
\centering\includegraphics[width=8.5cm,height=6.7cm]{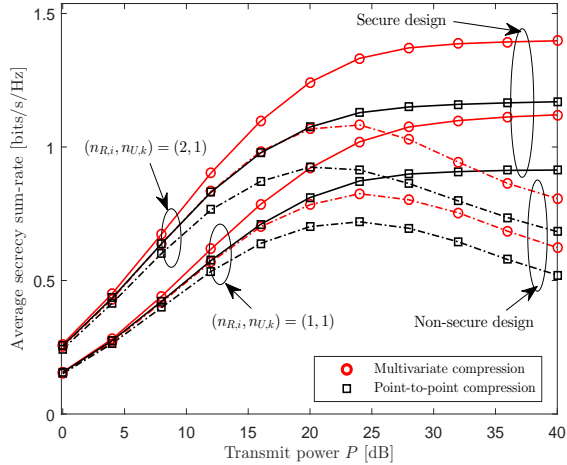}

\caption{\label{fig:graph-as-SNR}Average secrecy sum-rate versus the transmission
power $P$ for the downlink of a C-RAN with $N_{R}=2$, $N_{U}=3$,
$n_{U,k}=1$ and $C=2$ bits/s/Hz.}
\end{figure}

In Fig. \ref{fig:graph-as-SNR}, we plot the average secrecy sum-rate
versus the transmission power $P$ for the downlink of a C-RAN system
with $N_{R}=2$, $N_{U}=3$, $n_{U,k}=1$ and $C=2$ bits/s/Hz. It
is seen that the secure design significantly outperforms the non-secure
design. Also, it is observed that multivariate compression yields
a significant performance gain that is increasing with the transmission
power $P$. This is because the impact of the quantization noise $\mathbf{H}_{k}\mathbf{q}$
in (\ref{eq:received-signal-rewritten}) compared to the additive
noise $\mathbf{z}_{k}$ is more significant when the SNR is large
at the UE side. It is noted that the secrecy sum-rate of the secure
design saturates to a finite level at high-SNR, since, in this regime,
the performance is limited by the power of the quantization noise
that does not decrease with the SNR. Moreover, the performance of
the non-secure design is degraded in the high-SNR regime due to the
enhanced decodability of the messages of the unintended UEs. We also
observe that, comparing the curves with $(n_{R,i},n_{U,k})=(1,1)$
and $(n_{R,i},n_{U,k})=(2,1)$ suggests that increasing the number
of RU antennas results in improved performance since the additional
excessive antennas can be leveraged to obtain beamforming gains.

\begin{figure}
\centering\includegraphics[width=8.5cm,height=6.7cm]{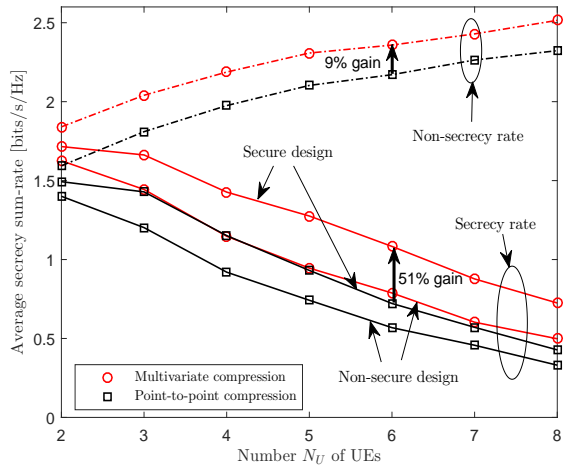}

\caption{\label{fig:graph-as-numUE}Average secrecy sum-rate and non-secrecy
sum-rate versus the number $N_{U}$ of UEs for the downlink of a C-RAN
with $N_{R}=3$, $n_{R,i}=n_{U,k}=1$, $P=20$ dB and $C=1$ bits/s/Hz.}
\end{figure}

Fig. \ref{fig:graph-as-numUE} plots the average secrecy sum-rate,
along with the standard ``non-secrecy'' sum-rate obtained without
imposing secrecy constraints, versus the number $N_{U}$ of UEs for
the downlink of C-RAN with $N_{R}=3$, $n_{R,i}=n_{U,k}=1$, $P=20$
dB and $C=1$ bits/s/Hz. The ``non-secrecy'' rates are obtained
by plugging the precoding and quantization noise covariance matrices
of the non-secure design into the weighted sum $\sum_{k\in\mathcal{N}_{U}}w_{k}R_{k}^{\prime}$
of rates $R_{k}^{\prime}=I(\mathbf{s}_{k};\mathbf{y}_{k})$ that are
achievable without secrecy constraints \cite[Eq. (13)]{Park-et-al:TSP13}.
We observe that, unlike the non-secrecy rate, as the number $N_{U}$
of UEs increases, the secrecy rate of all schemes becomes worse due
to the increased number $N_{U}-1$ of eavesdroppers on each message
$M_{k}$. However, the proposed secure design based on multivariate
compression provides a significant benefit in mitigating the impairments
from the eavesdroppers compared to the other schemes. In particular,
for $N_{U}=6$, multivariate compression yields a 51\% sum-rate gain
under the secrecy constraint while about 9\% is gained without secrecy
constraint. This demonstrates the additional role of artificial noise
that quantization noise shaping plays under secrecy constraint.

\section{Conclusions\label{sec:Conclusions}}

This work has proposed to leverage the quantization noise that is
inevitably added by fronthaul compression in a C-RAN downlink as ``artificial''
noise for enhancing the rate achievable under secrecy constraints.
To this end, we have investigated the application of multivariate
fronthaul quantization/compression at the CU in order to control the
statistics of the quantization noise across all the outgoing fronthaul
links. We have formulated the joint optimization problem of the precoding
and quantization noise covariance matrices for maximizing the weighted
sum of secrecy rates of the UEs subject to the per-RU fronthaul capacity
and power constraints. Numerical results were provided to verify the
effectiveness of multivariate compression in enhancing secret communication.

\end{document}